\documentclass[11pt]{article}
\usepackage{lmodern}
\usepackage[utf8]{inputenc}
\usepackage{color}
\usepackage{graphicx}
\usepackage{xspace}
\usepackage{multirow}
\usepackage{url}
\usepackage{amsmath,amssymb,amsfonts,amsthm}
\usepackage{microtype,mathtools,mathrsfs}
\usepackage{times}
\usepackage{hyperref}%                hyperlinks
\hypersetup{pdfstartview=XYZ}%        default zoom

\renewcommand{\baselinestretch}{1.2}
\usepackage{amsmath,amssymb,amsfonts}
\usepackage[lined]{algorithm2e}

\usepackage{multirow}
\usepackage{microtype,mathtools,mathrsfs}

\newtheorem{theorem}{Theorem}[section]

\newtheorem{claim}[theorem]{Claim}

\newtheorem{corollary}[theorem]{Corollary}

\newtheorem{observation}[theorem]{Observation}

\newcommand{\Natural}[0]{\mathbf{N}}

\def\cO{{\cal O}}

\def\cA{{\cal A}}

\def\cP{{\cal P}}

\def\Prob{\mbox{\rm\tt Pr}}  % notation taken from Motwani-Raghavan book

\newenvironment{smallitemize} {
  \begin{list}{$\bullet$} {\setlength{\parsep}{0pt}
\setlength{\itemsep}{0pt}} } { \end{list} }

\usepackage{times}

\begin{document}
\title{\bf Memory Lower Bounds for Randomized Collaborative Search and Applications to Biology}
\author{Ofer Feinerman\thanks{
The Louis and Ida Rich Career Development Chair, The Weizmann Institute of Science, Rehovot, Israel.
E-mail: {\tt feinermanofer@gmail.com}.
Supported by the Israel Science Foundation (grant 1694/10).}\\Weizmann Institute of Science \and
Amos Korman\thanks{
CNRS and University Paris Diderot, Paris, France.
E-mail: {\tt \,amos.korman@liafa.jussieu.fr}.
Supported in part by the ANR projects DISPLEXITY and PROSE, and by the INRIA project GANG.}\\ CNRS and Univ. Paris Diderot}
\date{}
\maketitle

\begin{abstract}
Initial knowledge regarding group size can be crucial for collective performance.
We study this relation in the context of the {\em Ants Nearby Treasure Search (ANTS)} problem \cite{FKLS}, which models natural cooperative foraging behavior such as that performed by ants around their nest.
In this problem,  $k$ (probabilistic) agents, initially placed at some central location,
collectively search for a treasure on the two-dimensional grid. The treasure is placed at a target location
by an adversary and the goal is to find it as fast as possible as a function of both $k$ and $D$, where $D$ is the (unknown) distance
between the central location and the target. It is easy to see that $T=\Omega(D+D^2/k)$ time units are necessary for finding the treasure. Recently, it has been established
that $O(T)$ time is sufficient if the agents know their total number $k$ (or a constant approximation of it), and enough memory bits are available at their disposal \cite{FKLS}.
In this paper, we establish lower
bounds on the agent memory size required for achieving
certain running time performances. To the best our knowledge, these bounds are the first non-trivial lower bounds for the memory size of probabilistic searchers. For example,  for every given positive constant $\epsilon$, terminating the search by time
$O(\log^{1-\epsilon} k \cdot T)$ requires agents to use  $\Omega(\log\log k)$ memory  bits.
Such distributed computing bounds may provide a novel, strong tool for the investigation of complex biological systems.

%\bigskip

%\paragraph*{\bf Keywords:} search algorithms; memory bounds; advice complexity; mobile robots; speed-up; cow-path problem; online algorithms; quorum sensing; social insects; central place foraging, ants.
\end{abstract}

\section{Introduction}\label{sec:introduction}

\noindent{\bf Background and Motivation:}
Individuals in biological groups assemble in groups that allow them, among other things, to monitor and react to relatively large environments.
For this, individuals typically disperse over length scales that are much larger than those required for communication. Thus, collecting knowledge regarding  larger areas  dictates a dispersion that may come    at the price of group coordination and efficient information sharing. A possible solution involves the use of designated, localized areas where individuals convene to share information and from which they then disperse to interact with the environment. Indeed, there are numerous examples for such convention areas in the biological world. Cells of the immune system undergo a collective activation, differentiation and maturation process away from the site of infection and within compact lymph nodes \cite{JTWS}. Birds are known to travel long distances from their feeding grounds to communal sleeping area where they were shown to share information regarding food availability \cite{Zahavi71,FDP}.  Here we focus, on a third example, that of collective {\em central place foraging} \cite{OP79,HM85} where a group of animals leave a central location (e.g., a nest) to which they then retrieve collected food items.  Here as well, the localized nest area enables efficient communication. Ants, for example, were shown to share information within their nest regarding food availability and quality outside it \cite{LBF,Gordon02}. This information is then used as a means of regulating foraging efforts.

%\noindent{\bf quorum sensing: }
One piece of information that may be available to a localized group is its {\em size}. Group size may be used to reach collective decisions which then affect the subsequent behavioral repertoire of the individuals. The most prevalent example is that of {\em quorum sensing};  a binary estimate of group size or density. Such threshold measurements are exhibited by a multitude of biological systems such as bacteria \cite{SMB}, amoeba \cite{TGK}, T-cells of the immune system \cite{BMA,FJTC} and social insects \cite{Pratt05}. The quorum sensing process constitutes a first decision step that may lead to cell differentiation and divergent courses of action. Going beyond quorum sensing: there are evidences for higher resolution estimates of group size in, for example, wild dogs where multiple hunting tactics are employed in correlation with increasing numbers of participating individuals \cite{T92}.

Here, we focus on the potential benefits of estimating group size in the context of collective central place foraging.
%This central location could be a food storage area, a nest where offspring are reared or simply a sheltered or familiar environment.
Ants, for example, engage in this behavior in a cooperative manner -  individuals search for food items around the nest and share any findings. Clearly, due to competition  and other time constrains, food items must be found relatively fast. Furthermore, finding food not only fast but also in proximity to the central location holds numerous advantages at both the search and the retrieval stages.  Such advantages include, for example, decreasing predation risk \cite{K80}, and increasing the rate of food collection once a large quantity of food is found \cite{OP79,HM85}. Intuitively, the problem at hand is distributing searchers within bounded areas around the nest while minimizing overlaps.
 Ants may possibly exchange information inside the nest, however, once they are out, minimizing search overlaps decreases the rate of communication via pairwise interaction, in some species, to a seemingly negligible degree  \cite{HM85}.

%We have previously introduced efficient search algorithms that scale well with the number of participating agents. Further, we have shown that for such collective searches to approach optimality, agents must possess some knowledge of the size of the search group. As detailed above, a central place allows for a preliminary stage in which the group may assess its size. In this work we are interested in the effect of the coarseness of this estimation on subsequent performance.
It was previously shown that the efficiency of collective central place foraging may be
enhanced by initial knowledge regarding group size \cite{FKLS}.
More specifically, that paper introduces the {\em Ants Nearby Treasure Search (ANTS)}  problem, which models  the aforementioned central place foraging setting.
In this problem,  $k$ (probabilistic) agents, initially placed at some central location,
collectively search for a treasure in the two-dimensional grid. The treasure is placed at a target location
by an adversary and the goal is to find it as fast as possible as a function of both $k$ and $D$, where $D$ is the (unknown) distance
between the central location and the target. Once the agents initiate the search they cannot communicate between themselves.  Based on volume considerations, it is an easy observation that the expected running time of any algorithm is $\Omega(D+D^2/k)$. It was established in \cite{FKLS} that the knowledge of a constant approximation of $k$ allows the agents to find the treasure
in asymptotically optimal expected time, namely, $O(D+D^2/k)$. On the other hand,  the lack of any information  of $k$ prevents them from reaching expected time that is higher than optimal by a factor slightly larger than $O(\log k)$. That work also establishes lower bounds on the competitiveness of the algorithm in the particular case where some given approximation to $k$ is available to all nodes.

%In other words, without performing any type of communication prior to the search and gaining information about $k$, the running time cannot be $O(\log k)$-competitive. (lower bound) from PODC?)

%In this work we are interested in the effect of the coarseness of this information on the running time.  More specifically,

In this work, we simulate the initial step of information sharing (e.g., regarding group size) within the nest  by using the abstract framework of {\em advice} (see, e.g., \cite{CFIKP08, FIP06,FKL10}).
 That is, we model the preliminary process for gaining knowledge about  $k$ (e.g., at the central location)  by means of an  {\em oracle} that assigns advice  to agents. To measure the amount of information accessible to agents, we analyze the {\em advice size}, that is, the maximum number of bits used in an advice. Since we are mainly interested in lower bounds on the advice size required to achieve a given competitive ratio, we apply a liberal
  approach and assume a highly powerful oracle. More specifically, even though it is supposed to model a distributed (probabilistic) process, we assume that the oracle is a centralized probabilistic algorithm (almost unlimited in its computational power) that can assign each agent with a different advice.  Note that, in particular, by considering identifiers as part of the advice, our model allows to relax the assumption that all agents are identical and to allow  agents to be of  several types. Indeed, in the context of ants, it has been established that ants on their first foraging bouts execute different  protocols than those that are more experienced~\cite{WMZ04}.

%To the best of our knowledge, this paper is the first paper to consider  probabilistic oracles.

% We ask questions such as: "can one
% design both an oracle that assigns short advices and a  search algorithm that uses the advices at agents and still be time efficient?"

The main technical results of this paper deal with lower bounds on the advice size.
For example, with the terminology of advice, Feinerman et al. \cite{FKLS} showed that advice of size $O(\log\log k)$ bits is sufficient to obtain an $O(1)$-competitive algorithm.
We prove that this bound is tight. In fact, we show a much stronger result, that is, that advice of size $\Omega(\log\log k)$ is necessary
even for achieving competitiveness which is as large as $O(\log^{1-\epsilon} k)$, for every given positive constant $\epsilon$.
 On the other extremity, we show that $\Omega(\log\log\log k)$ bits of advice are necessary for
being $O(\log k)$-competitive, and that this bound is tight. In addition, we exhibit  lower bounds on the corresponding advice size for a range of  intermediate competitivenesses.

Observe that  the advice size bounds from below the number of memory bits used by an agent, as this amount
of bits in required merely  for storing some initial information. In general, from a purely theoretical point of view, analyzing the memory required for efficient search is a central theme in computer science \cite{Rein08,Ro08}, and is typically considered to be difficult. To the best of our knowledge, the current paper is the first paper establishing non-trivial lower bounds for the memory of randomized searching agents with respect to given time constrains.

From a high level perspective, we hope to illustrate that  distributed computing  can potentially provide a novel and efficient methodology for the study of highly complex, cooperative biological ensembles. Indeed, if experiments that 
comply with our setting reveal that the ants' search is time efficient, in the sense detailed above, then our theoretical results can provide some insight  on the memory ants use for this task. A detailed discussion of this approach is given in Section~\ref{sec:conclusion}.

\paragraph{Our results:}
%Our results are summarized in the table below.
%Table~\ref{tab:results}.
%\subsubsection{Technical contribution}
%\noindent{\bf Technical contribution:}
The main technical results deal with lower bounds on the advice size. Our first result is perhaps the most surprising one. It says not only  that $\Omega(\log\log k)$ bits of advice are required to obtain an $O(1)$-competitive algorithm, but that
roughly this amount is necessary
even for achieving competitiveness which is as large as $O(\log^{1-\epsilon} k)$, for every given positive constant $\epsilon$.
This result should be put in contrast to the fact   that with no advice at all, one can obtain a search algorithm whose competitiveness is slightly higher than logarithmic \cite{FKLS}.

%This means that in order to reduce the memory to $o(\log\log k)$, one must pay a very large penalty in the competitiveness-- a penalty that is less than polynomially close to the one corresponding to the case where no advice is given whatsoever.

\begin{theorem}\label{lower-loglog}
There is no search algorithm that is $O(\log^{1-\epsilon} k)$-competitive for some fixed positive~$\epsilon$, using advice of size $o(\log\log k)$.
\end{theorem}

On the other extremity,  we show that %even though competitiveness that is slightly more than logarithmic is possible to achieve with no advice at all~\cite{FKLS}, we show that
% the main lower bound result of \cite{FKLS} implies that with zero bits of advice, one cannot construct an $O(\log k)$-competitive algorithm. We strengthen this result by showing that
 $\Omega(\log\log\log k)$ bits of advice are necessary for constructing
an $O(\log k)$-competitive algorithm, and we prove that this bound on the advice is in fact tight.

\begin{theorem}\label{thm:logloglog-lower}
There is no $O(\log k)$-competitive search algorithm, using advice of size $\log\log\log k-\omega(1)$. On the other hand, there exists an $O(\log k)$-competitive search algorithm using advice of size $\log\log\log k+O(1)$.
\end{theorem}

Finally, we also exhibit  lower bounds for the corresponding advice size for a range of  intermediate competitivenesses.
%Specifically, we show that for every given positive constant $\epsilon$, advice of size $\Omega( \log^{\epsilon}\log k)$ bits is necessary to achieve competitiveness of $\frac{\log k}{2^{\log^\epsilon\log k}}$.

\begin{theorem}
Consider  a $\Phi(k)$-competitive search algorithm using advice of size $\Psi(k)$. %Assume that $\Phi(\cdot)$ is relatively-slow.
Then, $\Phi(k)=\Omega(\log k / 2^{\Psi(k)})$, or in other words,  ${\Psi(k)}=\log\log k - \log \Phi(k) - O(1)$. In particular, if $\Phi(k)=\frac{\log k}{2^{\log^\epsilon\log k}}$, then $\Psi(k)= \log^{\epsilon}\log k-O(1)$.
\end{theorem}

Our results on the advice complexity are summarized in Table 1.   As mentioned, our lower bounds on the advice size are also lower bounds on the memory size of agents. 
%To the best of our knowledge, this paper is the first to exhibit non-trivial lower bounds for the memory size of probabilistic agents in the context of search problems.
%In addition, this paper is the first to consider and analyze probabilistic advice.

%\paragraph{Outline:} In Section~\ref{sec:preliminaries}, we formally define the model and the complexity measures. Essentially, all our lower bounds on the advice size follow as corollaries of our main technical result, stated in Theorem~\ref {main-lower}. Due to lack of space, the proof of the upper bound given in Theorem~\ref{thm:logloglog-lower} is deferred to the Appendix.
%In Section~\ref{sec:conclusion}, we discuss our results and propose a  novel approach for a scientific study.

%Observe also, that a lower bound of $\Psi$ on the memory size, imply that the number of states that an agent
%has is at least $2^\Psi$.

{\renewcommand{\arraystretch}{0.8}
\renewcommand{\tabcolsep}{0.2cm}
\begin{table*}\label{tab:results}
\begin{center}
%\scalebox{0.8}
{
\bgroup\large
\begin{tabular}{|l|l|l|}
\hline
& Competitiveness
& Advice size    \\

\hline
%\multirow{2}{3cm}
{Tight bound }
 & $O(1)$    & $\Theta(\log\log k)$\\

\hline
%\multirow{2}{3cm}
{Tight bound }
 & $O(\log^{1-\epsilon}  k)~~~~~~~0<\epsilon < 1$    & $\Theta(\log\log k)$\\
    % \multirow{2}{5cm}{}

%\cline{2-4}

%\hline
%Lower bound  & $\Phi(k)$ & $ \log\log k  -\log \Phi(k)-O(1)$\\

\hline
Lower bound  & $\log k/ 2^{\log^\epsilon\log k}~~~0<\epsilon < 1$  & $ \log^{\epsilon}\log k-O(1)$ \\

\hline
Tight bound  & $O(\log k)$ & $\log \log\log  k+ \Theta(1)$\\
\hline
Upper bound  ~~\cite{FKLS}& $O(\log^{1+\epsilon} k)$ & zero \\

\hline
\end{tabular}
\egroup
}
\end{center}
\caption{Bounds on the advice for given competitiveness}
\end{table*}

%\subsection
\paragraph{Related Work:}
Our current work falls within the framework of natural algorithms, a recent attempt to
study biological phenomena from an algorithmic
perspective~\cite{AAB+11, BMV12, Cha09,FKLS}.
%The most relevant work to ours is ~\cite{FKLS} which introduces the ANTS problem, and establishes
%a tight bound for the penalty that must be paid in the running time if no information about $k$ is available. That work also establishes lower bounds on the competitiveness of the algorithm in the particular case where a given approximation to $k$ is available to all nodes. Our work abstracts the information known to agents using the framework of advice, and establishes lower bounds on their size for given time constrains.

The notion of advice is central in computer science.
%(in fact, checking membership
%in NP-languages can be viewed as computing with advice).
 In particular,
the concept of advice  and its impact on various computations
has recently found various applications in distributed computing. In this context, the main measure used is the advice size.
 It is
for instance analyzed in frameworks such as  proof labeling \cite{KK07,KKP10}, broadcast \cite{FIP06}, local computation of MST \cite{FKL10}, graph coloring  \cite{FGLP07}
and  graph searching by a single robot \cite{CFIKP08}. Very recently, it has also been investigated in the context of online algorithms \cite{BKKK11, EFKR11}.
%We note that to the best of our knowledge, in all previous works on advice, it was assumed that the oracle assigning the advice is deterministic.
%This paper is thus the first to consider probabilistic oracles.

Collective search is a classical problem that has been extensively studied in different contexts (for a more comprehensive summary refer to \cite{FKLS}). Social foraging theory \cite{GC00} and central place foraging typically deal with optimal resource exploitation strategies between competing or cooperating individuals. Actual collective search trajectories of non-communicating agents have been studied in the physics literature (e.g.,  \cite{R06,BH97}). Reynolds \cite{R06} achieves optimal speed up through overlap reduction which is obtained by sending searchers on near -straight disjoint lines to infinity. This must come at the expense of finding proximal treasures. Harkness and Maroudas
 \cite{HM85} combined field experiments with computer simulations of a semi-random collective search and suggest substantial speed ups as  group size increases. The collective search problem has further been studied from an engineering perspective (e.g., \cite{PYP01}). In this case, the communication between agents  (robots) or their computational abilities are typically unrestricted. These works put no emphasis on finding nearby treasures fast. Further, there is typically    (with the exception of \cite{HM85}) no reference to group size or its knowledge by the agents.

In the theory of computer science,
the exploration of graphs using mobile agents (or robots) is a central question. (For a more detailed survey refer to e.g.,~\cite{FKLS,FDPS10}.) Most graph exploration research in concerned  with the case
of a single deterministic  agent exploring a finite graph, see for example \cite{AH00, BFRSV,DP02, DFKP02, GPRZ07, PP99, Rein08}. The more complex situation of multiple identical deterministic agents was studied in~\cite{AB97, FDPS10, FDPS11,FGKP04}. 
In general, one of the main challenges in search problems is the establishment of memory  bounds. % on the memory of agents.
For example, the question of whether a single agent can explore all finite undirected graphs using logarithmic memory was open for a long time; answering it to the affirmative \cite{Rein08} established an equality between the classes of languages SL and L. As another example, it was proved in \cite{Ro08} that no finite set of constant memory agents  can explore
all graphs. 
%To the best of our knowledge, the current paper is the first paper establishing non-trivial lower
%bounds for the memory of randomized searching agents with respect to given time constrains.

% tradeoffs
%between the efficiency of the search (typically, the time to find the treasure) and the
%memory of an agent has been an important theme in computer science.

The simplest (and most studied) probabilistic search algorithm is  the random
walk.
%is a natural candidate:  it is extremely simple, uses no memory, and
%trivially self-stabilizes.
In particular, several studies  analyzing the speed-up measure for  $k$-random walkers
have recently been published.
In these papers, a speed-up of $\Omega(k)$ is established for various finite graph families,
including, e.g., expenders and random graphs~\cite{AAKKLT, ES11, CFR09}.
%Alon et al.
%, developed several bounds that are based on
%the quotient between the cover time and maximum hitting times. Their
%technique
%shows a speed-up of $\Omega(k)$ for many graph families, e.g., for expenders.
%In  \cite{ES11}, Elsasser and  Sauerwald
%present a new lower bound on the speed-up that depends on the
%mixing-time. It gives a speed-up of $\Omega(k)$ on various graphs, even if $k$
%is as large as $n$. In \cite{CFR09}, Cooper et al.,  show a linear speed-up in $k$ for random graphs.
In contrast, for the two-dimensional $n$-node grid, as long as $k$ is polynomial in $n$, the speed up is only logarithmic in~$k$.
The situation with infinite grids is even  worse.
Specifically, though the $k$-random walkers find the treasure with probability one, the expected (hitting) time  is infinite.

% The more complex case with k-random walkers has recently
%gained attention in the literature [AAKKLT11, ES11]. In particular, for the finite grid, as
%long as k is polynomial in n, the speed-up is only logarithmic in k. The situation with
%infinite grids is even worse. Specifically, though the k-random walks will find the treasure
%with probability one, the expected time to find the treasure will be infinite.
%

%From a theoretical point of view, graph exploration and searching for treasures are very related tasks. These tasks are central questions in  computer science.
%How to explore graphs using mobile agents (sometimes called robots) is a central  question in  computer science.
%Algorithms for graph exploration by deterministic agents  have been intensely studied.

 Evaluating the running time as a function of $D$, the distance to the treasure, was studied in the context of the cow-path problem.
%is the following problem: a Cow (agents) is standing on the crossroad (refer as the origin) with $m$ infinite paths leading into unknown territory. On one of the paths there is a treasure (Meadow Field) at distance $D$ from the origin. The goal of the searchers is to find the treasure. The cow has no knowledge were the treasure is, or what is the distance $D$. The question is how to find the treasure effectively.
Specifically,  it was established in \cite{BCR91} that the competitive ratio for deterministically finding~a~point on the real line is nine, and that in the two-dimensional grid,  the spiral search algorithm is optimal up to lower order terms. Several other varients where studied in  \cite{DFG2006,KRT96, KSW86, LS01}. In particular,
in \cite{LS01},  the cow-path problem was extended by considering $k$ agents. However, in contrast to our setting, the agents they consider have unique identities, and the goal is achieved by (centrally) specifying a different path for each of the $k$ agents.

%Broadly speaking, the difficulty with deterministic agents comes from symmetry breaking issues.  When it comes to probabilistic agents, such
%difficulties can sometimes be solved quite efficiently. The natural algorithm consisting of an agent performing a \emph{random walk} is a good example.
%In fact, in some respects, random walks are very efficient searching algorithms, at least, for finite graphs. They are extremely  simple, they use no memory,
%and they trivially self-stabilize.
%Unfortunately, however, random walks turn out to be inefficient in an infinite grid. Specifically, if the grid is infinite 2-dimensional then the expected hitting time is infinite.
%

%In the case were the grid is finite having $n$ nodes it is convenient to measure the performance by speed up that the $k$-walks gives. In this case it was proved in \cite{AAKKLT} that the cover time of $k$-random walks is always bigger than $\frac{n}{\log(k)}$. This is proof is also implies that the hitting time is bigger than $\frac{n}{\log(k)}$.

The question of how important it is for individual processors to know their total number has recently been addressed in the context of locality.  Generally speaking, it has been observed that for several
classical local computation tasks, knowing the number of processors is not essential~\cite{KSV11}. On the other hand, in the context of  local distributed decision, some evidence exist that such knowledge is crucial for non-deterministic verification~\cite{FKP11}.

%

%Ð We find the correct order of the speed-up for any value of 1  k  n
%on hypercubes, random graphs and expanders. For d-dimensional
%torus graphs (d > 2), our bounds are tight up to a factor of O(log n).
%Ð Our findings also reveal a surprisingly sharp dichotomy on several
%graphs (including d-dim. torus and hypercubes): up to a certain
%threshold the speed-up is k, while there is no additional speed-up
%above the threshold.

%It is known that two cooperating agents can learn exactly any strongly-connected directed graph with indistinguishable nodes in expected polynomial time~\cite{BS94}.

%Recently there are several works that try to extend the classical searcher problems in the single Turing machine model to multi searchers problems in the context of distributed computing, see for example \cite{AAKKLT}. We believe that the right measure for the performance of the distributed algorithm is the speed-up it gains over the single searcher algorithm. This approach has two basic advantages, it bridges the classical results for one searcher $k=1$ to the general distributed case $k$ searcher problem. Anther advantage of the speed-up approach is that we can neglect low order terms. This happens since we are looking at the ratio of time it takes to find the treasure for a single searcher and the time it takes to find the treasure for $k$-searchers. The low order terms disappear when we look at the asymptotes.

\section{Preliminaries}\label{sec:preliminaries}
{\bf General setting:}
We consider the {\em Ants Nearby Treasure Search (ANTS)}  problem initially introduced in \cite{FKLS}. In this {\em central place} searching problem, $k$ mobile {\em agents}  are searching for a {\em treasure} on the two-dimensional plane.
The agents are probabilistic mobile machines (robots).
They are identical, that is, all agents execute the same protocol $\cP$.
Each agent has some limited field of view, i.e., each agent can see its surrounding up to a distance of  some $\varepsilon>0$. Hence, for simplicity, instead of considering the two-dimensional plane, we assume that the agents  are actually walking on the integer two-dimensional infinite grid $G=\mathbb{Z}^2$ (they can traverse an edge of the grid in both directions). The search is central place, that is, all $k$ agents initiate the search from some central node $s\in G$, called the {\em source}. Before the search is initiated,
an adversary locates the treasure  at some node $t\in G$, referred to as the {\em target} node.
%This aims to capture, for example, the situation where at sun rise, ants go out of their nest searching for
%nearby food (e.g., a large dead insect) that was left somewhere during the night.
 Once the search is initiated, the agents cannot communicate among themselves.
 We denote by $D$ the (Manhattan)  distance between the source node and the target, i.e., $D=d_G(s,t)$. It is important to note that the agents have no a priori information about the location of $t$ or about $D$.
 We say that the agents {\em find} the treasure when one of the agents visits the target node~$t$.
 The goal of the agents it to find the treasure as fast as possible as a function of both $D$ and $k$.

Since we are mainly interested in lower bounds, we assume a very liberal setting. In particular,
 we do not restrict neither the computational power nor  the navigation capabilities of agents. Moreover, we put no restrictions on the internal storage used for navigation\footnote{On the other hand, we note that for constructing upper bounds, the algorithms we consider use simple procedures that can be implemented using relatively little resources.
 For example, with respect to navigation, the constructions only  assume  the ability to perform four basic procedures, specifically: (1) choose a
direction uniformly at random, (2) walk in a ``straight line'' to a prescribed distance and direction, (3) perform a {\em spiral search}  around a given node (see, e.g., \cite{BCR91}),
%\footnote{The spiral search is a particular deterministic search algorithm that starts at a node $v$ and enables the agent  to visit % all nodes at distance  $\Omega(\sqrt{x})$ from $v$ by traversing $x$ edges, for every integer $x$  (see, e.g., \cite{BCR91}). For %our purposes, since we are interested with asymptotic results only, we can replace this atomic
%navigation protocol with any procedure that guarantees such a property. For simplicity, in what follows, we assume that for any %integer $x$, the spiral search of length $x$ starting at a node $v$ visits all nodes at distance at most $\sqrt{x}/2$ from $v$},
and (4) return to the source node.}. \\

\noindent{\bf Oracles and Advice:}
We would like to model  the situation in which before the search actually starts, some  initial communication may be made  between the  agents at the source node.
%(Note, however, that once an agent starts the search it cannot communicate with other agents anymore.)
In reality, this preliminary communication may be quite limited. This may be because of  difficulties in the communication that are inherent to the agents or the environment, e.g., due to faults or limited memory, or because of asynchrony issues regarding the different starting times of the search, or simply because agents are identical and it may be difficult for agents to distinguish one agent from the other. Nevertheless,  we consider a very liberal setting in which this preliminary communication  is almost unrestricted.

 More specifically, we consider a centralized algorithm called {\em oracle} that assigns advices to agents in a preliminary stage. The oracle, denoted by $\cO$, is a probabilistic\footnote{It is not clear whether  or not  a probabilistic oracle is strictly more powerful than a deterministic one.  Indeed, the oracle assigning the advice is unaware of $D$, and may thus potentially use the randomization to reduce the size of the advices by
 balancing between the efficiency of the search for small values of $D$ and larger values.} centralized algorithm that receives as input a set of $k$ agents and assigns an {\em advice} to each of the $k$ agents.
We  assume that the oracle may use a different protocol for each $k$; given $k$, the randomized algorithm used for assigning the advices to the $k$~agents is denoted by $\cO_k$. Furthermore,  the oracle may assign a different advice to each agent\footnote{We note that even though we consider a very liberal setting, and allow a very powerful oracle, the oracles we use for our upper bounds constructions
are very simple and rely on much weaker assumptions. Indeed, these oracles are not only deterministic but also assign the same advice to each of the $k$ agents.}.
Observe, this definition of an oracle allows it to simulate almost any reasonable preliminary communication between the agents\footnote{For example, it can simulate to following very liberal setting. Assume that in the preprocessing stage,   the $k$ agents are organized in a clique topology, and
that each agent can send a separate message to each other agent. Furthermore, even though the agents are identical, in this preprocessing stage, let us assume that agents can distinguish the messages received from different  agents, and that each of the $k$ agents may use a different probabilistic protocol for this preliminary  communication.  In addition, no restriction is made neither on  the memory and computation capabilities of agents nor on the preprocessing  time, that is, the preprocessing stage takes finite, yet unlimited, time.}.

%Therefore, in what follows, we always assume that prior to the actual search, a randomized oracle assigns the advice to agents .

 It is important to stress that even though all agents execute the same searching protocol, they may start the search with different advices. Hence, since their searching protocol may rely on the content of this initial advice,
 agents with different advices may behave differently. Another important remark  concerns the fact that some part of the advices may be used for encoding (not necessarily disjoint)
 %\footnote{Observe, in this case, if all $k$ identifiers are pairwise disjoint then the data-structures must use at least $\log k$ bits.}
  identifiers. That is, assumptions regarding the settings in which not all agents are identical and there are several types of agents can be captured by our setting of advice.
 %Indeed, in the context of ants, it has been established (Wehner) that younger  ants  execute different  protocols than older ones.

To summarize, a {\em search algorithm} is a pair $\langle \cP,\cO\rangle$ consisting of a randomized searching protocol $\cP$ and randomized oracle $\cO=\{\cO_k\}_{k\in \Natural}$. Given
$k$ agents, the randomized oracle $\cO_{k}$ assigns a separate advice to each of the given agents. Subsequently, all agents initiate the actual search by letting each of the  agents  execute protocol $\cP$ and using the corresponding advice as input to $\cP$. Once the search is initiated, the agents cannot communicate among themselves.

%We stress, while searching,
%all agents  execute the same randomized algorithm, nevertheless, this randomized algorithm receives the advice as input from the oracle, and hence,
%we may expect agents with different advices to possibly behave differently.
%the oracle assigns advices to the $k$ agents. Each such advice is encoded using at most $\Psi(k)$ bits. Hence,
%we view the advice assignment as a function $\cO_{k\in \Natural}:  \{1,\cdots, k\} \rightarrow \{1,\cdots, 2^{\Psi(k)}\}$.
%(We interpret $\cO_{k}(j)$ as the random variable indicating the advice given to the $j$'s agent assuming we start with $k$ agents.)

%A {\em uniform} algorithm is an algorithm in which no information regarding $k$ is available to agents, i.e., $\Psi(k)=0$.
%(The term uniform is chosen to stress that
%agents execute the same algorithm regardless of their number, see, e.g., \cite{KSV11}.)

%\noindent{\bf Advice complexity and Memory complexity:}
Consider  an oracle $\cO$.
%, let  $\cO_k$ be a random algorithm used by the oracle assuming we have $k$ agents. Recall, for each agent, $\cO_k$ may use a different distribution of advices.
%Let $\cO_k(j)$ denote the  random algorithm used by the oracle for assigning the advice to the $j$'th agent.
Given $k$, let $\Psi_{\cO}(k)$ denote the maximum number of bits devoted for encoding the advice of an agent, taken over all coin tosses of $\cO_k$, and over the $k$ agents. In other words, $\Psi_{\cO}(k)$ is the minimum number of bits necessary for encoding the advice, assuming the number of agents is $k$.  Note that
$\Psi_{\cO}(k)$ also bounds from below the number of memory bits of an agent required by the search algorithm $\langle \cP,\cO\rangle$, assuming that the number of agents is $k$.  The function $\Psi_{\cO}(\cdot)$ is called the {\em advice size} function of oracle  $\cO$.
%Whenever convenient, in what follows, we may think of an advice as some number (or state) in the range $\{1,\cdots, 2^{\Psi_{\cO}(k)}\}$.
(When the context is clear, we may omit the subscript $\cO$ from $\Psi_{\cO}(\cdot)$ and simply use $\Psi(\cdot)$ instead.)\\

%Note,  given $k$, let  $\cO_k$ is random algorithm that assigns each agent  an advice in the range $\{1,\cdots, 2^{\Psi_{\cO}(k)}\}$.
 %The function $\Psi_{\cO}(k)$ is called the advice size function of oracle $\cO$.

%Observe, a lower bound on the size of the advice is also a lower bound on the {\em memory size} of an agent, which is required merely
%for storing the output of the preprocessing stage, i.e., the initial data-structure required for the search. In other words, even without considering the memory of agents used for the actual navigation, the lower bounds on the advice we establish imply that so many bits of memory must be used initially by the agents. \\

\noindent{\bf Time complexity:}
When measuring the time to find the treasure, we assume that all internal computations are performed in zero time.
For the simplicity of presentation, we assume that the movements of agents are synchronized, that is, each edge traversal is performed in precisely one unit of time. Indeed, this assumption can easily be removed if we measure the time according to the slowest edge-traversal. We also assume that all agents start the search simultaneously at the same time. This assumption can also be easily removed by starting to count the time when the last agent initiates the search.

The {\em expected running time} of a search algorithm  $\cA:=\langle \cP,\cO\rangle$ is  the expected time until at least one of the agents finds the treasure.
The expectation  is defined with respect to the coin tosses made by the (probabilistic) oracle $\cO$ assigning the advices to the agents, as well as the subsequent coin tosses made by the agents executing $\cP$. We denote the expected running time of an algorithm $\cA$ by $\tau=\tau_{\cA}(D,k)$.
In fact, for our lower bound to hold, it is sufficient to assume that the probability that the treasure is found by time $2\tau$ is at least $1/2$. By Markov inequality, this assumption is indeed weaker than the assumption that the expected running time is $\tau$.

Note that if an agent knows $D$, then it can potentially find the treasure in time $O(D)$, by
walking to a distance $D$ in some direction, and then performing a circle around the source of radius $D$ (assuming, of course, that its navigation abilities enable it  to perform such a circle). On the other hand, with the absence of knowledge about $D$, an agent can find the treasure in time $O(D^2)$ by performing a spiral search  around the source (see, e.g., \cite{BCR91}). The following observation imply that $\Omega(D+D^2/k)$ is a lower bound on the expected running time of any search algorithm. The proof is straightforward and can be found in  \cite{FKLS}.

\begin{observation}\label{simple-lower}
%When considering $k$ agents,
%it is easy to see
%\footnote{To see why, consider a search algorithm whose expected running time is $T$. Clearly,  $T\geq D$, because it takes $D$ time to merely reach the treasure.
%Assume, towards contradiction, that $T<D^2/4k$.  In any execution of $\cA$, by time $2T$, the $k$ agents can visit a total of at most $2Tk<D^2/2$ nodes. Hence,
% by time  $2T$, more than half of the nodes in $B_D:=\{u\mid 1\leq d(u)\leq D\}$
%were not visited. Therefore, there must exist a node
%$u\in B_D$  such that the probability that $u$ is visited by time $2T$ (by at least one of the agents) is less than $1/2$. If the adversary locates the treasure at $u$ then the expected time to find the treasure is strictly greater than~$T$, which contradicts the assumption.}
%that
The expected running time of any algorithm is $\Omega(D+D^2/k)$, even if the number of agents~$k$ is known to all agents.
%, and even if we relax the model and allow agents to freely communicate between each other at all times.
\end{observation}

%It follows from \cite{FKLS} that if $k$ is known to agents then there exists a search algorithm whose expected running  time is asymptotically optimal, namely, $O(D+D^2/k)$.
We evaluate the time performance of an algorithm  with respect to the lower bound given by Observation~\ref{simple-lower}.
Formally,  let $\Phi(k)$ be a function of $k$. A search  algorithm $\cA:=\langle \cP,\cO\rangle$ is called {\em $\Phi(k)$-competitive}~if
$$\tau_{\cA} (D,k)\leq \Phi(k)\cdot (D+D^2/k),$$ for every integers $k$ and $D$.
Our goal is establish connections between the {\em size} of the advice, namely $\Psi(k)$, and the competitiveness  $\Phi(k)$ of the search algorithm. \\

\noindent{\bf More definitions:}
The {\em distance} between two nodes  $u,v\in G$, denoted $d(u,v)$, is simply the Manhattan  distance between them, i.e., the number of edges on the shortest path connecting $u$ and $v$ in the grid~$G$.  For a node $u$, let $d(u):=d(u,s)$ denote the distance between $u$ and the source node.
Hence, $D=d(t)$.
%Let $B(r)$ denote the ball centered at the source $s$ with radius $r$, formally, $B(r)=\{v\in G: d(v)\leq r\} $.

\section{Lower Bounds on the Advice}\label{sec:lower}
The  theorem below generalizes Theorem 4.1 in \cite{FKLS}, taking into account the notion of advice. %The proof of the theorem contains the main technical contribution of our paper.  
All our lower bound results follow as corollaries of this theorem.
Note that for the theorem to be meaningful we are interested in advice size  whose order of magnitude is less than $\log\log k$. Indeed, if $\Psi(k)=\log\log k$, then one can encode a 2-approximation of $k$ in each advice, and obtain an optimal result, that is, an $O(1)$-competitive algorithm (see \cite{FKLS}).
%Recall, also, that for every positive $\epsilon$, there exists a uniform algorithm (i.e., with no advice) that is $\log^{1+\epsilon} k$-competitive.
%Fix some function $\Psi(k)$. For simplicity, think for now, that $\Psi(k)=\lceil\log\log\log k\rceil$. For an integer $x$, let $k_1(x)$ (respectively, $k_2(x)$) denote the smallest (resp. largest)
%integer such that $\Psi(k_1(x))=x$ (resp., $\Psi(k_2(x))=x$). In particular, for every integer $k_1(x)\leq k\leq k_2(x)$, we have $\Psi(k)=x$. Note, if $f$ is slow (which we assume), then $k_1(x)$ and $k_2(x)$ can %potentially be far apart. More precisely,

Before stating the theorem, we need the following definition. A non-decreasing function $\Phi(x)$ is called {\em relatively-slow} if $\Phi(x)$  is sublinear (i.e., $\Phi(x)=o(x)$) and if
there exist two positive constants $c_1$ and $c_2<2$ such that when
%the following properties are satisfied: (1) there exists a constant $c$ such that
restricted to $x>c_1$, we have %the function $x/\Phi(x)$ is increasing and diverging, and (2) there exists a positive constant   $c_2<2$, such that
$\Phi(2x)<c_2\cdot \Phi(x)$. Note that this definition captures many natural sublinear functions\footnote{For example, note that the functions of the form $\alpha_0+\alpha_1\log^{\beta_1} x+\alpha_2\log^{\beta_2}\log x+ \alpha_3 2^{\log^{\beta_3 }\log x} \log x+\alpha_4\log^{\beta_4} x\log^{\beta_5}\log x$, (for
non-negative constants $\alpha_i$ and $\beta_i$, $i=1,2,3,4,5$ such that $\sum_{i=1}^4\alpha_i>0$)
 are all relatively-slow.}. %Given a relatively-slow  function $\Phi(x)$, let $x_0$ be the smallest integer $x_o>c$ such that $x_0/\Phi(x_0)>1$.

\begin{theorem}\label{main-lower}
Consider  a $\Phi(k)$-competitive search algorithm using advice of size $\Psi(k)$. Assume that $\Phi(\cdot)$ is relatively-slow  and that $\Psi(\cdot)$ is non-decreasing. Then there exists some constant $x'$, such that for every $k>2^x{'}$,
the sum $\sum_{i=x'}^{\log k} \frac{1}{\Phi(2^i)\cdot 2^{\Psi(k)}}$ is at most some fixed constant.
\end{theorem}

%We say that a  non-decreasing function $f:\Natural \rightarrow \Natural$ is {\em sub-exponential} if there exists a constant $c<2$ such that
%$f(x+1) <c \cdot  f(x)$, for every $x\in \Natural$.
%\begin{theorem}\label{th:lowerbound}
%Consider a sub-exponential function $f$, and let $g(k)=\sum_{j=1}^{\log k} 1/f(j)$.  Let $h(k)=2^{\Psi(k)}$ and assume that $\lim g(k)/h(k)=\infty$.
%There is no uniform search
%algorithm that is $O(f(\log k))$-competitive using advice of size $\Psi(k)$.
%\end{theorem}

%$\Phi(x)=f(\log x)$ so $\Phi(2^i)=f(i)$.
%Take $f(x)=x/\log x$ so $f(\log k)=\log k/\log\log k$. Now $g(k)=\sum_{j=1}^{\log k} 1/f(j)=\sum_{j=1}^{\log k} \log j/j=\Omega(\log^2 \log k)$

\begin{proof}
Consider a search algorithm  with advice size $\Psi(k)$ and competitiveness $\Phi'(k)$, where $\Phi'(\cdot)$ is relatively-slow.
By definition,  the expected running time is less
than $\tau(D,k)=(D+D^2/k)\cdot\Phi'(k)$.  Note, for $k\leq D$, we have
$\tau(D,k)\leq \frac{D^2\Phi(k)}{k}$, where $\Phi(k)=2\Phi'(k)$, and $\Phi(\cdot)$ is relatively-slow. Let $c_1$ be the constant promised by the fact that  $\Phi$ is relatively-slow. Let $x_0>c_1$ be sufficiently large so that $x_0$ is  a power of 2, and  for every $x>x_0$, we have $\Phi(x)<x$ (recall, $\Phi$ is sublinear).
%Assume
% that $\Phi(k)=O(\log k)$.

Fix an integer $T>x_0^2$.
In the remaining  of the proof, we assume that the treasure is placed somewhere at  distance $D:=2T+1$. Note, this means, in particular,
that by time $2T$ the treasure has not been found yet.
%Next, for every integer $i\in [\log x_0,\frac{1}{2}\log T]$, set
%$
%k_i=2^i\quad\text{and}\quad d_i=\sqrt{\frac{T\cdot k_i}{ \Phi(k_i)}}.
%$

Fix an integer $i$ in $[\log x_0,\frac{1}{2}\log T]$,
set $$d_i=\sqrt{\frac{T\cdot k_i}{ \Phi(k_i)}},$$
and let  $B(d_i):=\{v\in G: d(v)\leq d_i\} $ denote the ball of radius $d_i$ around the source node.
We consider now the case where the algorithm is executed with $k_i:=2^i$ agents (using the corresponding advices  given by the oracle $\cO_{k_i}$). For every set of nodes $S\subseteq B(d_i)$, let $\chi_i(S)$ denote the random variable
indicating the number of nodes in $S$ that were visited by at least one of the
$k_i$ agents by time $2T$.  (For short, for a singleton node $u$, we write $\chi_i(u)$ instead of  $\chi_i(\{u\})$.) Note, the value of  $\chi_i(S)$ depends on the values of the coins tosses made by the oracle for assigning the advices as well as on the values of the coins tossed by  the $k_i$ agents.
Now, define the ring $R_i :=B(d_i)\setminus B(d_{i-1})$.
%The proof of the following claim is deferred to Appendix~\ref{app:claim}.

\begin{claim}\label{claim:exp}
For each integer $i\in[\log x_0,\frac{1}{2}\log T]$, we have $\mathbf{E}(\chi_i(R_i))=\Omega(d_i^2)$.
\end{claim}

%\section{Proof of Claim~\ref{claim:exp}}\label{app:claim}
To see why the claim holds, note that by the properties of $\Phi$, and from the fact that $2^i\leq \sqrt{T}$, we get  that $k_i\leq d_i$, and therefore, $\tau(d_i,k_i)\leq \frac{d_i^2\Phi(k_i)}{k_i}= T$. It follows that for each node $u\in B(d_i)$, we have $\tau(d(u),k_i)\leq T$, and hence, the
% expected time to visit $u$ is at most  $\tau(d_i,k_i)\leq{d_i^2\Phi(k_i)}/{k_i}= T$.
%Thus, by Markov's inequality, t
%The
probability that $u$ is visited by time $2T$ is at least $1/2$, that is,
$\mathbf{Pr}(\chi_i(u)= 1)\geq 1/2$. Hence, $\mathbf{E}(\chi_i(u))\geq 1/2.$
Now, by linearity of expectation, $\mathbf{E}(\chi_i(R_i))=\sum_{u\in R_i} \mathbf{E}(\chi_i(u))\geq |R_i|/2.$
Consequently, by time $2T$, the expected number of nodes in $R_i$ that  are visited by the $k_i$ agents is
$\Omega(|R_i|)=\Omega\left(d_{i-1}(d_i-d_{i-1})\right)=\Omega\left(\frac{T\cdot k_i}{\Phi(k_{i-1})}\cdot \left( \sqrt{ \frac{2\Phi(k_{i-1})}{\Phi(k_{i})}} - 1 \right)\right)=\Omega \left(\frac{T\cdot k_i}{\Phi(k_i)}\right)=\Omega(d_i^2)$,
where the second equality follows from the fact that $d_i=d_{i-1}\cdot
\sqrt{\frac{2\Phi(k_{i-1})}{\Phi(k_{i})}}$, and the third equality follows from the fact that $\Phi(\cdot)$ is relatively-slow.
This establishes the claim. \\
%\qed

Note that for each $i\in [\log x_0+1,\frac{1}{2}\log T]$, the advice given by the oracle
to any of the $k_i$ agents must use at most $\Psi(k_i)\leq \Psi(\sqrt{T})$ bits. In other words, for each of these $k_i$ agents, each advice is some integer whose value is at most $2^{\Psi(\sqrt{T})}$.

%Recall that $\cO_{k_i}(j)$ is the advice that the oracle gives to the $j$'th agent assuming we start with $k_i$ agents. It follows that
Let $W(j, i)$ denote the random variable indicating the number of nodes in $R_i$ visited by the $j$'th agent by time $2T$, assuming that the total number of agents is  $k_i$.
%(Note, the value of  $W(j, i)$ depends on the values of the coins tosses made by the oracle for assigning the advices as well as on the values of the coins tossed by  the $k_i$ agents.)
By Claim~\ref{claim:exp}, for every integer $i\in[\log x_0+1,\frac{1}{2}\log T]$, we have:
$$
\mathbf{E}\left(\sum_{j=1}^{k_i} W(j,i)\right)~\geq~\mathbf{E}(\chi_i(R_i))~=~ \Omega(d_i^2).
$$
By linearity of expectation, it follows that for every  integer $i\in[\log x_0+1,\frac{1}{2}\log T]$, there exists an integer $j\in \{1,2,\cdots, k_i\}$   for which
$$
\mathbf{E}(W(j,i))=\Omega(d_i^2/k_i)=\Omega(T/\Phi(k_i)).
$$

%Recall that $\cO_{k_i}(j)$ is the advice that the oracle gives to the $j$'th agent assuming we start with $k_i$ agents. It follows that
%$M(\cO_{k_i}(j),i)$ is the random variable indicating the number of nodes in $R_i$ visited by the $j$'th agent  (among the $k_i$) by time $2T$.
%Hence, by Claim~\ref{claim:exp}, for every integer $i\in[\log x_0+1,\frac{1}{2}\log T]$, we have:
%$$
%\mathbf{E}\left(\sum_{j=1}^{k_i} M(\cO_{k_i}(j),i)\right)~\geq~\mathbf{E}(\chi_i(R_i))~=~ \Omega(d_i^2).
%$$
%By linearity of expectation, it follows that for every  integer $i\in[\log x_0+1,\frac{1}{2}\log T]$, there exists an agent (among the $k_i$ agents) that received an advice $\alpha_i$ for which
%$$
%\mathbf{E}(M(\alpha_i,i))=\Omega(d_i^2/k_i)=\Omega(T/\Phi(k_i)).
%$$

Now, for each advice in the relevant range, that is, for each $a\in\{1,\cdots, 2^{\Psi(\sqrt{T})}\}$, let $M(a,i)$ denote the random variable indicating the number of nodes in $R_i$ that an agent with advice $a$ visits by time $2T$. Note, the value of $M(a,i)$ depends only  on the values of the coin tosses made by the agent.
On the other hand,
note that the value of $W(j,i)$ depends on the results of the coin tosses made by the oracle assigning the advice, and the results of the coin tosses made by the agent that uses the assigned advice.
Recall,
the oracle may assign an advice to agent $j$ according to a distribution that is different than the distributions used for other agents. However, regardless of the distribution used by the oracle for agent $j$, it must be the case that there exists an advice $a_i
 \in\{1,\cdots, 2^{\Psi(\sqrt{T})}\}$, for which $
\mathbf{E}(M(a_i
,i))\geq \mathbf{E}(W(j,i))$. Hence, we obtain:
$$
\mathbf{E}(M(a_i
,i))=\Omega(T/\Phi(k_i)).
$$

%(amos: Im almost 100 percent sure that the above claim is correct. By i still need a formal argument to be convinced.)

Let $A=\{a_i
\mid  i\in[\log x_0+1,\frac{1}{2}\log T]\}$.
Consider now an ``imaginary'' scenario\footnote{The scenario is called imaginary, because,  instead of letting the oracle assign the advice
for the agents, we impose a particular advice to each agent, and let the agents perform the search with our advices. Note, even though such a scenario cannot occur by the definition of the model, each individual agent with advice $a$ cannot distinguish this case from the case that the number of agents was some $k'$ and the oracle assigned it the advice $a$.} in which  we execute the search algorithm with $|A|$ agents, each having a different advice in~$A$. That is, for each advice $a\in A$, we have a different agent executing the algorithm using advice $a$.
 For every set $S$ of nodes, let $\hat{\chi}(S)$ denote the random variable
indicating the number of nodes in $S$ that were visited by at least one of these
$|A|$ agents (in the ``imaginary'' scenario) by time $2T$. Let $\hat{\chi}:=\hat{\chi}(G)$ denote the random variable
indicating the total number of nodes that were visited by at least one of these
agents by time $2T$.

By definition, for each $i\in[\log x_0+1,\frac{1}{2}\log T]$, the expected number of nodes in $R_i$ visited by at least one of these $|A|$ agents is
$$\mathbf{E}(\hat{\chi}(R_i))\geq \mathbf{E}(M(a_i
,i))=\Omega(T/\Phi(k_i)).$$
%For each $a\in\{1,\cdots, 2^{\Psi(\sqrt{T})}\}$, let $$I_a:=\{i\in \{ 1,2, \cdots, \log T/2\} \mid  a_i  =a \}.$$
%The following observation follows from the pigeon hole principle.
%\begin{observation}
%There exists $a\in\{1,\cdots, 2^{\Psi(\sqrt{T})}\}$, such that $|I_a| \geq \log T/2^{\Psi(\sqrt{T})+1}$.
%\end{observation}
%%To se why the observation holds, assume towards contradiction, that for each $j\in\{1,2,\cdots, 2^{\Psi(K)}\}$,  we have $|I_j| < \log k/2^{\Psi(K)+1}$.
%For the remaining of the proof, we fix the advice  $a$ as promised by the observation. By definition, for every $i\in I_a$, we have $a_i=a$, and hence,
%$$
%E(M(a,i))=\Omega(T/\Phi(k_i)).
%$
Since the sets $R_i$ are pairwise disjoint, the linearity of expectation
implies that the expected number of nodes covered by these agents by time $2T$ is
$$
\mathbf{E}(\hat{\chi})~\geq~\sum_{i=x_0+1}^{\frac{1}{2}\log
T}  \mathbf{E}(\hat{\chi}(R_i))   ~=~  \Omega\left(\sum_{i=x_0+1}^{\frac{1}{2}\log
T}\frac{T}{\Phi(k_i)} \right) $$  $$=T\cdot \Omega\left(\sum_{i=x_0+1}^{\frac{1}{2}\log
T}\frac{1}{\Phi(2^i)}\right).
$$
Recall that $A$ is included in
$\{1,\cdots, 2^{\Psi(\sqrt{T})}\}$. Hence, once more by linearity of expectation,    there must exist an advice $\hat{a}\in A$, such that the expected number of nodes that an agent with advice $\hat{a}$ visits by time $2T$ is
$$
T\cdot \Omega\left(\sum_{i=x_0+1}^{\frac{1}{2}\log
T}\frac{1}{\Phi(2^i)\cdot 2^{\Psi(\sqrt{T})}} \right).
$$
Since each agent may visit at most one node in one unit of time, it follows that,  for every $T$ large enough, the sum $$\sum_{i=x_0+1}^{\frac{1}{2}\log T} 1/\Phi(2^i)\cdot 2^{\Psi(\sqrt{T})}$$ is at most some fixed constant.  The proof of the theorem now follows
by replacing the variable $T$ with $T^2$.
\end{proof}
%\bigskip
%Our first corollary generalizes the lower bound from \cite{FKLS}, assuming an approximation $k^\epsilon$.
\begin{corollary}
Consider  a $\Phi(k)$-competitive search algorithm using advice of size $\Psi(k)$. Assume that $\Phi(\cdot)$ is relatively-slow.
Then, $\Phi(k)=\Omega(\log k / 2^{\Psi(k)})$, or in other words,  ${\Psi(k)}=\log\log k - \log \Phi(k) - O(1)$.

\end{corollary}

\begin{proof}
Theorem~\ref{main-lower} says that for every $k$, we have $$\frac{1}{2^{\Psi(k)}}\sum_{i=1}^{\log k} \frac{1}{\Phi(2^i)}=O(1).$$ On the other hand, since $\Phi$ is non-decreasing, we have
$$\sum_{i=1}^{\log k} \frac{1}{\Phi(2^i)}\geq \frac{\log k}{\Phi(k)}.$$
Hence,
$$ \frac{\log k}{2^{\Psi(k)}\cdot\Phi(k)}=O(1).$$
The corollary follows.
\end{proof}

The following corollary follows directly from the previous one.

\begin{corollary}
Let $\epsilon<1$ be a positive constant. Consider  a $\frac{\log k}{2^{\log^\epsilon\log k}}$-competitive search algorithm using advice of size $\Psi(k)$. Then
 $\Psi(k)= \log^{\epsilon}\log k-O(1)$.
\end{corollary}

%Recall also, that there exists an $O(1)$-competitive algorithm using advice of size $\log\log k$.
Our next corollary implies that even though $O(\log\log k)$ bits of advice are sufficient for obtaining $O(1)$-competitiveness, roughly this
 amount of advice  is necessary even for achieving  relatively large competitiveness.

%\bigskip
\begin{corollary}\label{cor:main}
There is no $O(\log^{1-\epsilon} k)$-competitive search algorithm  for some positive constant $\epsilon$, using advice of size $\Psi(k)= \epsilon\log\log k -\omega(1)$.
\end{corollary}

\begin{proof}
Assume that the competitiveness is $$\Phi(k)=O(\log^{1-\epsilon} k).$$ Then,  $$\sum_{i=1}^{\log k} \frac{1}{\Phi(2^i)\cdot 2^{\Psi(k)}}=\Omega\left(\frac{\log^{\epsilon} k}{2^{\Psi(k)}}\right).$$
According to Theorem~\ref{main-lower}, this sum must converge,  and hence, we cannot have $\Psi(k)=\epsilon\log\log k-\omega(1)$.
\end{proof}
%\bigskip
\begin{corollary}\label{cor:logk}
There is no $O(\log k)$-competitive search algorithm, using advice of size $\log\log\log k-\omega(1)$.
\end{corollary}

\begin{proof}
Assume that the competitiveness is  $\Phi(k)=O(\log k)$. Since $\Phi(2^i)=O(i)$, we have
%Let $\psi(k)=\sum_{i=1}^{\log k} 1/\Phi(2^i)$.
%Hence,
$$\sum_{i=1}^{\log k} 1/\Phi(2^i)=  \sum_{i=1}^{\log k} 1/ i= \Omega(\log\log k).$$
According to Theorem~\ref{main-lower}, $$\frac{1}{2^{\Psi(k)}}\sum_{i=1}^{\log k} 1/\Phi(2^i)=\Omega(\log\log k/2^{\Psi(k)})$$ must converge as $k$ goes to infinity.
 In particular, we cannot have $\Psi(k)=\log\log\log k-\omega(1)$.
 %, since that would mean that $2^{\Psi(k)}= o(\log\log k)$.
\end{proof}

\section{Upper Bound}\label{sec:upper-bound}
The lower bound on the advice size given in Corollary~\ref{cor:main} is tight, as $O(\log\log k)$ bits of advice are sufficient to obtain an $O(1)$-competitive search algorithm.
 To further illustrate the power of Theorem~\ref{main-lower}, we now
show that the lower bound  mentioned in Corollary~\ref{cor:logk} is also tight. Theorem \ref{thm:logloglog-lower} follows by combining the theorem below and
Corollary~\ref{cor:logk}.
\begin{theorem}\label{th:lowerbound}
%Consider the {\em deterministic-uniform} advice of size $\Psi(k)$.
There exists  an $O(\log k)$-competitive algorithm using $\log\log\log k+O(1)$ bits of advice.
%Then there is no search algorithm that is $o(\log k)$-competitive.
\end{theorem}
\begin{proof}
%\noindent{\bf Proof.}
For each integer $i$, let $B_i:= \{u\,:\, d(u) \leq 2^i\}$.  Without loss of generality, we may assume that  $k$ is sufficiently large, specifically, $k\geq 4$.
Given $k$ agents, the oracle simply encodes the advice $O_k=\lfloor\log\log k\rfloor$ at each agent.
% (resp. $\cO_k=2^{\log^{\epsilon}\log k}$).
Note that since $k\geq 4$, we have $O_k\geq 1$. Observe also that the advice $O_k$ can be encoded using $\log\log\log k+O(1)$ bits.
 %(resp. $\log\log^{\epsilon} k$ bits).
% We now describe the algorithm that each agent with advice $\alpha$ performs.

We now explain the search protocol used by an agent having an advice $\alpha$. Let $K(\alpha)$ be the set of integers $k$ such that
$k$ is a power of 2 and $O
_k=\alpha$.
Let $g(\alpha)$ be the number of elements in $K(\alpha)$, i.e., $g(\alpha)=|K(\alpha)|$.
We enumerate the elements in $K(\alpha)$ from small to large, namely, $$K(\alpha)=\{ k_{1}(\alpha),k_{2}(\alpha),\cdots, k_{g(\alpha)}(\alpha)  \},$$ where
   $k_{\rho}(\alpha)$ is the $\rho$'s smallest integer in $K(\alpha)$, for $\rho\in\{1,2, \cdots, g(\alpha)\}$.
% (Resp. $2^{\log^{\epsilon}\log K_{max}}=2^{\log^{\epsilon}\log K_{min}}+1$, and hence, for $\epsilon=1/2$, ${\sqrt{\log\log K_{max}}}<\sqrt{{\log\log K_{min}}}+1$ .
%Thus
%$$
%\log\log K_{max} < \log\log K_{min} +1 +2\sqrt{\log\log K_{min}}
%$$
%
%$$
%\log(\log K_{max}/\log K_{min})<1 +2\sqrt{\log\log K_{min}}
%$$
%hence,
%${\log\log K_{max}}<2 {\log\log K_{min}}$)
%Given an advice $\alpha$ and
%an integer $\rho\in\{1,2, \cdots, g(\alpha)\}$,
%let $k_{\rho}(\alpha)$ be the $\rho$'th value of $w$ such that (1) $w=2^j$ for some $j$, and (2) $\cO_w=\alpha$.
Consider Algorithm $\cA$ described below.
% and let us analyze its performances.

\renewcommand{\baselinestretch}{1.1}
\begin{algorithm}
%\DontPrintSemicolon
\Begin{
Let $\alpha$ be the advice given to the agent. The agent first chooses an integer $\rho\in\{1,2, \cdots, g(\alpha)\}$ uniformly at random, and then performs the following double loop\;
\For(the \emph{stage} $j$ ){$j$ from $1$ to $\infty$}{%
  \For(the \emph{phase} $i$ ){$i$ from $1$ to $j$}{%
%     \For(the \emph{pulse} $r$ defined as follows){$r$ from $1$ to $K_a$}{%
  \begin{smallitemize}
  \item
 Let $t_i=2^{2i+{\alpha}+2}/k_\rho(\alpha)$.
   \item
Perform a spiral search (starting\\ at   the source) for time $t_i+2^i$.
 \item
Return to the source $s$
      \item
  Go to a node $u\in B_i$ chosen\\ uniformly  at random among the\\ nodes in $B_i$
\item
Perform a spiral search for time $t_i$.
\item
Return to the source $s$.
  \end{smallitemize}
%   }
  }
 }
}
\caption{The  algorithm $\cA$.}\label{alg:k-alg}
\end{algorithm}

Informally,  apart from having an oracle encoding the advice and letting each agent ``guess'' the number of agents using the advice, the algorithm differs structurally from the previous algorithms in \cite{FKLS} by the fact that in each phase $i$,
each agent first performs a spiral search around the source before ``jumping'' to a random node in $B_i$.

%Due to lack of space, its  proof is deferred to Appendix~\ref{app:upper}.
%\begin{lemma}\label{lem:upper}
%Algorithm $\cA$ is $O(\log k)$-competitive.
%\end{lemma}

Let us analyze the performances of algorithm $\cA$. Our goal is to show that $\cA$ is $O(\log k)$-competitive. We begin with the following observation.
\begin{observation}
$g(O_k)=O(\log k)$.
\end{observation}

To see why the observation holds, let $k_{\max}$ denote the maximum value such that $O_{k_{\max}}=O_k$, and let $k_{\min}$ denote the minimum value such that $O_{k_{\min}}=O_k$.  We know, $g(O_k)\leq \log k_{\max}$.
Since $O_{k_{\max}}=O_{k_{\min}}$, we have $\log\log k_{\max}<\log\log k_{\min}+1$, and hence $\log k_{\max}<2\log k_{\min}$. The observation follows, as $k_{\min}\leq k$.

Observe now that there exists $\rho^*\in\{1,2, \cdots, g(\alpha)\}$, such that  $k_{\rho^*}(\alpha)\leq k< 2k_{\rho^*}(\alpha)$. Let $N$ denote the random variable indicating the number of agents  that choose $\rho^*$.
Since each agent chooses an integer $\rho\in\{1,2, \cdots, g(\alpha)\}$ uniformly at random,
 then the expected value of $N$ is
$$\mathbf{E}(N)=k/g(\alpha)=\Omega( k/ \log k).$$
Let us now
 condition on the event that $N\geq \mathbf{E}(N)/2$. For the purposes of the proof, we consider for now only those $N$ agents.
Note, since these $N$ agents choose $\rho^*$, then they execute phase $i$ is Algorithm $\cA$ with $$t_i=\Theta(2^{2i}\cdot\frac{\log k}{k})=\Omega(2^{2i}/N).$$
Recall that $D$ denotes the distance from the treasure to the source. Let $s=\lceil\log D\rceil$. Fix a positive integer $\ell$ and consider the time $T_\ell$ until all the $N$ agents completed $\ell$
phases $i$ with $i\geq s$.
%Provided that $\alpha$ is sufficiently large,
Each time an agent performs phase $i$, the agent finds the
treasure (in the second spiral search it performs)  if the chosen node $u$ belongs to
the ball $B(v, \sqrt{t_i}/2)$ around  the node $v$ holding the treasure.
Note that at least some constant fraction of the ball $B(v, \sqrt{t_i}/2)$ is contained in $B_i$.
The probability of choosing a node $u$ in that fraction is thus
$$\Omega(|B(v,\sqrt{t_i}/2)|/|B_i|)=\Omega(\log k/k),$$ which is at least $\beta/N$ for
some positive constant $\beta$.
Thus, the probability that by time $T_\ell$ none of the $N$ agents finds the treasure (while
executing their respective $\ell$ phases~$i$) is at most
$(1-\beta/N)^{N\ell}$, which is at most $\gamma^{-\ell}$ for some
constant $\gamma$ greater than $1$.

For an integer $i$, let $\psi(i)$ be the time required (for one of the $N$ agents) to execute a phase $i$.
Note that $$\psi(i)=O(2^i+2^{2i}/N).$$ Hence, the time elapsed from the beginning of the algorithm until all the $N$ agents
complete stage $j_0$ for the first time is
$$\sum_{j=1}^{j_0}\sum_{i=1}^j \psi(i)
=O\left(\sum_{j=1}^{j_0} \left( 2^j+\sum_{i=1}^j 2^{2i}/N\right)\right)=$$
$$=O(2^{j_0}+2^{2j_0}/N).$$
 It follows that for any integer $\ell$, all the $N$ agents
complete their respective stages $s+\ell$ by time
$$\hat{T}(\ell)=O(2^{s+\ell}+2^{2(s+\ell)}/N).$$ Observe that by this time, all these $N$ agents
have completed at least $\ell^2/2$ phases $i$ with $i\geq s$.
Consequently, the probability that none of the $N$ agents finds the treasure by time
$\hat{T}(\ell)$  is at most
$\gamma^{-\ell^2/2}$. Hence, conditioning on the event that $N>\mathbf{E}(N)/2$,
the expected running time is at most
$$O\left(\sum_{\ell=1}^\infty
\hat{T}(l)\cdot \gamma^{-\ell^2/2} \right)=
O\left(\sum_{\ell=1}^\infty
\frac{2^{s+\ell}}{\gamma^{\ell^2/2}}+\frac{2^{2(s+\ell)}}{N\gamma^{\ell^2/2}}\right)$$
$$=O\left(2^{s}+2^{2s}/N\right)=O(D+D^2/N)=$$
$$=O(D+D^2/k)\cdot \log k.$$
%This establishes the theorem.
%\qed
On the other hand, by Chernoff  inequality, we have:
$$
\Prob[N\leq \mathbf{E}(N)/2]<e^{-E(N)/8}<e^{-\sqrt{k}}.
$$
Since each agent performs a spiral search of size $t_i+2^i$ around the source in each stage $i$, it follows that the treasure is found by time $O(D^2)$, with probability 1. Hence,  all together, the expected running time is
$$O\left(D^2\cdot e^{-\sqrt{k}}+(D+D^2/k)\cdot \log k\right)=$$
$$=O\left((D+D^2/k)\cdot \log k\right). $$
In other words, $\Phi(k)=O(\log k)$, as desired. 

Since the advice $O_k$ can be encoded using $\log\log\log k+O(1)$ bits, and since the competitiveness is $O(\log k)$,  the theorem follows. 
\end{proof}

\section{Conclusion and Discussion}\label{sec:conclusion}
As stated above, a central place allows for a preliminary stage in which a group of searchers may assess some knowledge about its size.
% From a biological perspective, very little is known about such processes.  
Our theoretical lower bounds on  advice size may enable us to relate group search performance to the extent of information sharing within the nest. Furthermore,
our  lower bounds on the memory size (or, alternatively, on the number of states), may provide some
evidence concerning the actual memory capacity of foraging ants.

A common problem, when studying a biological system is the complexity of the system and the huge number of parameters involved.
This raises the need for more concise descriptions and several alternatives have been explored. 
One tactic is reducing the parameter space.
This is done by
dividing the parameter space into critical and non-critical directions where changes
 in non-critical parameters do not affect overall system behavior \cite{FVDGA,GWCBMS}. A different approach involves the definitions of bounds which govern a biological systems. Bounds may originate from physics: the sensitivity of eyes is limited by quantum shot noise and that of biochemical  signaling pathways by noise originating from small number fluctuations \cite{Bialek87}. Information theory has also been used to formulate such bounds, for example, by bounding information transmission rates of sensory neurons \cite{RWB}. Note that the biological systems are confined by these bounds independently of any algorithms or parameters.

Our results are an attempt to draw non-trivial bounds on biological systems from the field of distributed computing. Such bounds are particularly interesting since they provide  not a single bound but  relations between key parameters. In our case these would be the memory capacity of an agent and collective search efficiency.
% Indeed, suppose experiments on living ants indicate that they solve the search problem relatively fast~\cite{HM85}; 
  Combining our memory lower bounds and measurements of search speed with varying numbers of searchers  would provide quantitative evidence regarding the 
 number of memory bits (or, alternatively, the number of states) used in the ants' search. In particular, this would help to understand the  ants' quorum sensing process, as this number of memory bits
 are required merely for representing the output of that process.
 %  resolution of the ants' quorum sensing: just how accurately they can count  their own number.
Obviously, to truly illustrate the concept, one  must give precise (non-asymptotic) bounds and  conduct a careful experiment. This is beyond the scope of this paper. 
Nevertheless, our results provide a ``proof-of-concept'' for such a methodology.

 Central place foraging with no communication is not uncommon in social insects as they search for food around their nest. Mid-search communication of  desert ants {\em Cataglyphys} and honeybees {\em Apis mellifera}, for example, is highly limited due to their dispersedness and lack of chemical trail markings. Note that, these insects do have an opportunity to interact amongst themselves before they leave their nest  as well as some capacity to assess their own number \cite{Pratt05}. Furthermore, both desert ants and bees possess many of the individual skills required for the behavioral patterns that are utilized in our upper bounds ~\cite{SW04,SZAT00,RSR,R08,RSMG07, WRSM81,HM85,WMZ04}. Such species are therefore natural candidates for the experiments described above.

\end{document}